\documentclass[journal, draftcls, onecolumn, 11pt]{IEEEtran}

\usepackage{amsmath, amssymb,amsthm}
\usepackage{mathtools}
\usepackage{inputenc, graphicx, verbatim,color, enumerate}
\usepackage{float}
\usepackage{tikz,pgfplots}
\usepackage{url}
\usepackage{xcolor,color}
\usepackage{cite}
\usepackage{makecell}

\newtheorem{theorem}{Theorem}
\newtheorem{proposition}{Proposition}
\newtheorem{lemma}{Lemma}
\newtheorem{corollary}{Corollary}
\newtheorem{definition}{Definition}
\newtheorem{example}{Example}

\newcommand{\F}{\mathbb{F}}

\newcommand{\R}{\mathbb{R}}
\newcommand{\Z}{\mathbb{Z}}

\newcommand{\GG}{\mathcal{G}}

\newcommand{\CC}{\mathcal{C}}
\newcommand{\RR}{\mathcal{R}}

\newcommand{\cl}{\mathrm{cl}}
\newcommand{\opt}{\mathrm{opt}}
\newcommand{\ie}{\emph{i.e.}}

\newcommand{\bm}[1]{\mathbf{#1}}

\newcommand\SmallMatrix[1]{{%
  \small\arraycolsep=0.6\arraycolsep\ensuremath{\begin{pmatrix}#1\end{pmatrix}}}}

\usepackage{hyperref}

\begin{document}
\title{Alphabet-Dependent Bounds for Linear Locally Repairable Codes Based on Residual Codes}

\author{%
    Matthias~Grezet,
    Ragnar~Freij-Hollanti, 
    Thomas~Westerb\"{a}ck, 
    and Camilla~Hollanti%
  \thanks{The work of M. Grezet and C. Hollanti was supported in part by the Academy of Finland, under grants \#276031, \#282938, and \#303819, and by the Technical University of Munich -- Institute for Advanced Study, funded by the German Excellence Initiative and the EU 7th Framework Programme under grant agreement \#291763, via a Hans Fischer Fellowship.}
  \thanks{R. Freij-Hollanti was supported by the German Research Foundation (Deutsche Forschungsgemeinschaft, DFG) under grant WA3907/1-1.}
  \thanks{M. Grezet, R. Freij-Hollanti, and C. Hollanti are with the Department of Mathematics and Systems Analysis, Aalto University, FI-00076 Aalto, Finland (e-mail: matthias.grezet@aalto.fi; ragnar.freij@aalto.fi; camilla.hollanti@aalto.fi).}
  \thanks{T. Westerb\"ack is with the Division of Applied Mathematics, UKK, M\"alardalen University, H\"ogskoleplan 1, Box 883, 721 23 V\"aster\r{a}s, Sweden (e-mail: thomas.westerback@mdh.se).}
  \thanks{0018-9448 \textcopyright \, 2019 IEEE. Personal use is permitted, but republication/redistribution requires IEEE permission. See http://www.ieee.org/publications\_standards/publications/rights/index.html for more information.} 
  }


\maketitle

\begin{abstract}
Locally repairable codes (LRCs) have gained significant interest for the design of large distributed storage systems as they allow a small number of erased nodes to be recovered by accessing only a few others. Several works have thus been carried out to understand the optimal rate--distance tradeoff, but only recently the size of the alphabet has been taken into account. In this paper, a novel definition of locality is proposed to keep track of the precise number of nodes required for a local repair when the repair sets do not yield MDS codes. Then, a new alphabet-dependent bound is derived, which applies both to the new definition and the initial definition of locality. The new bound is based on consecutive residual codes and intrinsically uses the Griesmer bound. A special case of the bound yields both the extension of the Cadambe-Mazumdar bound and the Singleton-type bound  for codes with locality $(r,\delta)$, implying that the new bound is at least as good as these bounds. Furthermore, an upper bound on the asymptotic rate--distance tradeoff of LRCs is derived, and yields the tightest known upper bound for large relative minimum distances. Achievability results are also provided by deriving the locality of the family of Simplex codes together with a few examples of optimal codes. 
\end{abstract}


%
\IEEEpeerreviewmaketitle



\section{Introduction}
%
%
%

In recent times, many service providers allow users to access and store data remotely to avoid overwhelming the limited storage capacity of single users. This leads naturally to the design of large distributed storage systems that reliably store data while minimizing the redundancy necessary to deal with server failures. 

The use of erasure-correcting codes together with network coding techniques for distributed storage systems, initiated in \cite{dimakis10}, has become popular since these so-called regenerating codes achieve the optimal tradeoff between the required repair bandwidth and storage overhead. For a standard erasure code of length $n$, dimension $k$ and minimum distance $d$, any $d-1$ failures can be repaired by contacting at most $k$ other nodes. In addition to this property, and at the cost of the failure tolerance, regenerating codes also enable efficient repair of failed nodes. This was long thought to be in contrast to the traditional maximum distance separable (MDS) codes that have to reconstruct the whole file in order to repair a single node. However,  \cite{guruswami17, rawat18} showed that this claim is not true in general, namely some MDS codes can also be efficiently repaired. Nevertheless,  the number of nodes contacted for repair can be a bottleneck for the system efficiency. To reduce the repair network traffic,  \cite{gopalan12} and later \cite{papailiopoulos12} introduced the notion of locality $r$ allowing the repair of a single failure to be done by contacting only $r$  nodes  with $r \ll k$. Erasure codes satisfying this requirement are called locally repairable (or recoverable) codes (LRCs).

A natural extension was presented in \cite{prakash12}, \cite{kamath13}  where the authors defined the locality $(r,\delta)_{i}$ for the information symbols to allow $\delta-1$ failures to be still corrected locally. This requirement was extended in \cite{prakash12} to all symbols without differentiating between the information symbols and the parity symbols. In this paper, we focus only on all-symbol locality and therefore drop the specification. A \emph{repair set} $R$ of an LRC is a set of coordinates such that any $\delta-1$ code symbols $c_{i}$ with $i \in R$ can be obtained from the remaining code symbols with indices in $R$. Other extensions of the locality property include codes with availability \cite{wang14b}, sequential repair of several erasures \cite{prakash14}, cooperative repair \cite{rawat15}, local repair on graphs \cite{mazumdar15} and many others. 

Abundant literature has been devoted to understanding the best possible parameters of LRCs and providing optimal constructions. The authors of \cite{gopalan12} gave the first tradeoff between the parameters $n,k,d$, and $r$ by showing that the minimum distance $d$ of an $(n,k,d,r)$-LRC with locality $r$ is bounded as follows:
\begin{equation}
\label{eq:Gopalan}
d \leq n-k-\left\lceil \frac{k}{r} \right\rceil +2.
\end{equation}

This bound was extended in \cite{prakash12} for any $(n,k,d,r,\delta)$-LRC with locality $(r,\delta)$ : 

\begin{equation}
\label{eq:Prakash}
d \leq n-k +1 - \left( \left\lceil \frac{k}{r} \right\rceil -1 \right) (\delta-1).
\end{equation}
The two bounds have been proven to be tight for large alphabet size with constructions provided in \cite{papailiopoulos12,prakash12, kamath13, rawat16, huang07, kamath13b, rawat14b, tamo16b, tamo14, goparaju14, ernvall16, westerback16}. Bounds for codes with availability were established in \cite{wang14b, rawat16, tamo16, huang16}. For a summary on various bounds for LRCs, see \cite{freij18}.  

The pioneering work done in \cite{cadambe15} improves on the bound \eqref{eq:Gopalan} by including a dependence on the alphabet size in the bound, that is, for any $(n,k,d,r)$-LRC over the alphabet $Q$ with $|Q|=q$, we have
\begin{equation}
\label{eq:CM}
k \leq \min\limits_{1 \leq t \leq n/(r+1)} \left\{ t r + k_{\opt}^{(q)}(n-t(r+1),d) \right\},
\end{equation}
where $k_{\opt}^{(q)}(n,d)$ is the maximal dimension of a code over $Q$ of length $n$ and minimum distance $d$. This has led to further construction of optimal LRCs over small alphabets, for example in \cite{silberstein18, zeh15}. 

Recently, the authors of \cite{agarwal18} proposed the first alphabet-dependent bounds on $(n,k,d,r,\delta)$-LRCs over $Q$ using an upper bound $B_{l-c}(r+\delta-1,\delta)$ on the cardinality of a code given its length $r+\delta -1$ and minimum distance $\delta$, with the extra requirement that the upper bound is a log-convex function on the length. The global bound is as follows:  
\begin{equation}
\label{eq:ABHMT}
k \leq \left( \left\lceil \frac{n-d+1}{r+\delta-1} \right\rceil +1 \right) \log_{q} B_{l-c}(r+\delta-1,\delta).
\end{equation}
A linear-programming bound for LRCs with locality $(r,\delta)$ was also derived in \cite{agarwal18} under the extra assumption that the repair sets are disjoint.

Finally, in \cite{grezet18}, the authors presented a Singleton-type bound for binary linear LRCs. This bound uses the local dimension of a repair set instead of the parameter $r$ and a more precise understanding of the intersection between two repair sets. As such, the work in this paper generalizes these two ideas.

While so far no distinction was made between non-linear and linear codes, the next results in this paragraph are only true for linear codes. In \cite{griesmer60}, Griesmer  proved the existence of a residual code for any binary linear code (over $\F_{2}$), \ie , a code obtained by a restriction with certain specific parameters. Griesmer then derived a lower bound on the length of the code given its dimension and minimum distance. The two results were later extended to an arbitrary field $\F_{q}$ in \cite{solomon65}. We present here the most general form. For any $[n,k,d]$ linear code $\CC$ over $\F_{q}$, there exists $\CC'$, a restriction of $\CC$ called the \emph{residual code of $\CC$}, such that $\CC'$ has parameters $[n-d, k-1, d' \geq \left\lceil d/q \right\rceil]$. By recursively taking residual codes, the authors of \cite{solomon65} obtained the following bound on the length $n$ of a linear code, known as the \emph{Griesmer bound}, and denoted here by $\GG(k,d)$:
\begin{equation}
\label{eq:Griesmer}
n \geq \sum\limits_{i=0}^{k-1} \left\lceil \frac{d}{q^{i}} \right\rceil =: \GG(k,d).
\end{equation}

\subsection{Our contributions}

In this paper, we first highlight the differences between the initial motivation for introducing the notion of locality in \cite{gopalan12, papailiopoulos12} and the definition of locality given in \cite{prakash12}, where the authors constrained the size of the repair sets. We show, through some examples, how the definition in \cite{prakash12} yields only a loose upper bound on the number of nodes contacted during the repair process when $\delta$ is larger than the alphabet size of the code. To remedy this, we introduce a new definition for locality called dimension-locality and compare it to the definition of locality in \cite{prakash12}. 

Then, we focus on linear LRCs and derive a new alphabet-dependent bound of the type of the bound \eqref{eq:CM} for linear codes with dimension-locality using the repair sets and chains of consecutive residual codes. Given the definition of dimension-locality, this bound also applies to linear LRCs with locality $(r,\delta)$ by using a weaker bound on the dimension of the local codes.  As a corollary of our results, we also derive a new Singleton-type bound that reflects better the actual dimension of the local codes. Furthermore, the new bound can be used to obtain a straightforward extension of the bound \eqref{eq:CM} for locality $(r,\delta)$ and the bound \eqref{eq:Prakash}, which shows that our bound is always at least as good as these bounds.

Next, we derive the asymptotic formulas of the new bound and the new Singleton-type bound when $n \to \infty$ to obtain the bounds on the tradeoff between the rate and the relative minimum distance. We also use these formulas to compare our bounds to the bounds in \cite{agarwal18} (Equation \eqref{eq:ABHMT} and \eqref{eq:ABHMT_asympt} here).
We show that there are cases where the new asymptotic Singleton-type bound is either better, equal, or worse than the asymptotic version of the bound \eqref{eq:Prakash}. The comparison with our main bound \eqref{eq:CMG_r} is more direct as we can prove that there is always an interval in the relative minimum distance where the new bound is strictly better than the bound \eqref{eq:Prakash}. Moreover, the improvement is quite significant since our bound benefits from the locality-unaware bounds on the rate--distance tradeoff. As an example, Figure \ref{fig:asympt_comp} displays the comparison between the known bounds and the new bound \eqref{eq:CMG_asympt} for linear LRCs with locality $(4,3)$ over the binary field, where we use the McEliece-Rodemich-Rumsey-Welch (MRRW) bound in \cite{McEliece77} as the intrinsic bound on the rate. Finally, we prove the achievability of the new bounds by studying the locality of the Simplex codes and providing a few optimal examples.

\begin{figure}
\centering
\includegraphics[height=6.5cm]{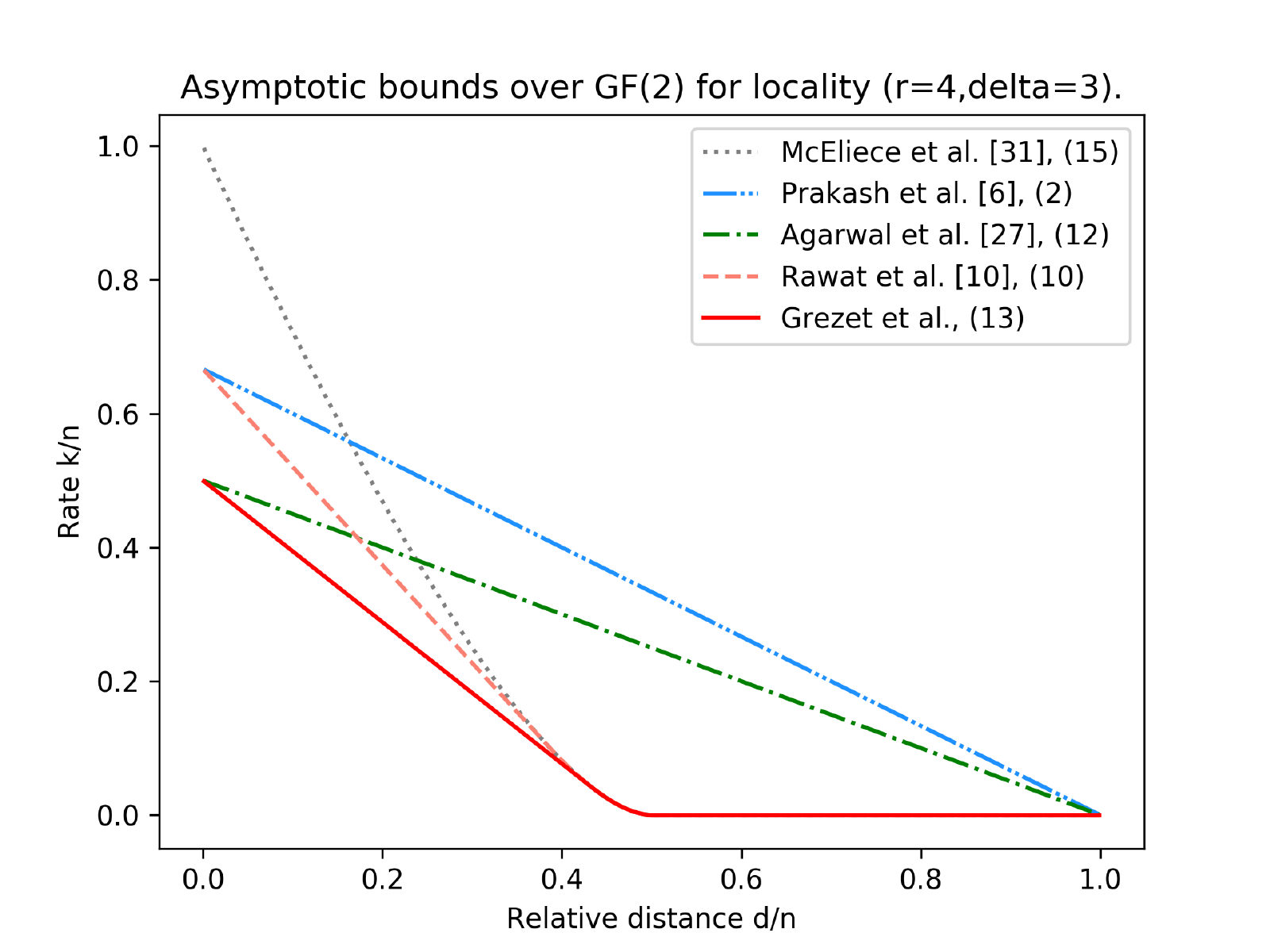}
\caption{Comparison of the upper bounds on the tradeoff between the rate $k/n$ and the relative minimum distance $d/n$ from \cite{prakash12, agarwal18, McEliece77}, and the bounds \eqref{eq:CM_rdelta} and \eqref{eq:CMG_asympt} over the binary field for large values of $n$ and fixed locality $(r=4,\delta=3)$.}
\label{fig:asympt_comp}
\end{figure}

The rest of the paper is organized as follows. In Section \ref{section:definitions}, we discuss the relation between the initial motivation that led to the introduction of locality and the definition given in \cite{prakash12}. Then, we define the notion of dimension-locality and compare it to the original definition of locality in \cite{prakash12}. In Section \ref{section:CMG}, we derive a new bound for linear LRCs with dimension-locality and extend it to linear LRCs with locality $(r,\delta)$. We derive also a new Singleton-type bound for these codes. In Section \ref{section:analysis}, we prove first that our bound is always as good as the straightforward extension of the bound \eqref{eq:CM} for locality $(r,\delta)$ and the bound \eqref{eq:Prakash}. Then, we derive the asymptotic formulas for the new bounds and use them for the comparison to the bound \eqref{eq:ABHMT}. While the comparison with the new Singleton-type bound depends on the parameters of the codes, we prove that our bound always improves on the bound \eqref{eq:ABHMT} for large relative minimum distances. Finally, in Section \ref{section:simplex}, we provide a family of codes that achieves our bounds by studying the locality of the Simplex codes.



\section{Mathematical preliminaries and locality revisited}
\label{section:definitions}

We denote the set $\{1,2,\ldots, n\}$ by $[n]$ and the set of all subsets of $[n]$ by $2^{[n]}$. The set of all positive integers including $0$ is denoted by $\Z_{+}$. For a length-$n$ vector $\mathbf{v}$ and a set $I \subseteq [n]$, the vector $\mathbf{v}_{I}$ denotes the restriction of the vector $\mathbf{v}$ to the coordinates in the set $I$. A linear code $\CC$ of length $n$, dimension $k$, and minimum distance $d$ is denoted by $[n,k,d]$ and a generator matrix for $\CC$ is $G_{\CC} = (\mathbf{g}_{1}, \ldots , \mathbf{g}_{n})$ where $\mathbf{g}_{i} \in \F_{q}^{k}$ is a column vector for $i \in [n]$. The number of codewords in $\CC$ is the cardinality of $\CC$, $|\CC| = q^{k}$. The \emph{shortening} of a code $\CC$ to the set of coordinates $I \subseteq [n]$ is defined by $$\CC/I = \{ \mathbf{c}_{[n]\setminus I} : \mathbf{c} \in \CC \text{ such that } c_{i}=0 \text{ for all } i \in I\},$$ and the \emph{restriction} of code $\CC$ to  $I$ is defined by $$\CC|_{I}=\{\mathbf{c}_{I} : \mathbf{c} \in \CC \}.$$ 
For convenience, we shortly call  the codes obtained by a restriction \emph{restricted codes}. Equivalently, the \emph{puncturing} of $\CC$ on the set $I$ is defined as the restriction of $\CC$ to the set $[n]-I$.
For an $[n,k,d]$ linear code $\CC$, if $\CC$ meets the Singleton bound, \ie , if $d=n-k+1$, then $\CC$ is called a maximum distance separable (MDS) code. 

To measure the dimension of restricted linear codes, or more generally, the amount of information contained in the restriction of arbitrary codes, we use the notion of an entropy function on the subsets $I \subseteq [n]$, where $n$ is the length of the code. We state it here for quasi-uniform codes over the alphabet $Q$. 

Quasi-uniform codes are a general class of error-correcting codes containing all the linear codes, group codes, and almost affine codes. More precisely, let $\CC$ be an arbitrary subset of $Q^{n}$. We can associate to $\CC$ a random vector $\bm{X}=(X_1, \ldots, X_n)$ with joint probability distribution
\[
P(\bm{X}=\bm{x})=\left\lbrace
\begin{array}{ll}
    1/ |\CC| & \text{if } \bm{x} \in \CC,\\
    0 & \text{otherwise.}
\end{array}
\right.
\]
Then, $\CC$ is a \emph{quasi-uniform} code if for all $A\subseteq [n]$, the restricted random vector $\bm{X}_{A}$ is uniformly distributed over its support $\lambda(\bm{X}_{A})=\{\bm{x}_{A} : P(\bm{X}_{A}=\bm{x}_{A})>0\}$. We refer to \cite{chan10} for more information about quasi-uniform codes and \cite{westerback18} for the entropy function on these codes. 

\begin{definition}
Let $\CC$ be a quasi-uniform code of length $n$ over the alphabet $Q$ and $I \subseteq [n]$. The entropy associated to $\CC$ is the function $H_{\CC} : [n] \to \R$ with
\[
H_{\CC}(I) = \frac{\log(|\CC|_{I}|)}{\log |Q|} = \frac{\log( |\{ \mathbf{c}_{I} : \mathbf{c} \in \CC \} | )}{\log |Q|}.
\]
\end{definition}

For ease of notation, if the underlying code of $H_{\CC}$ is clear, we drop the specification to $\CC$. For linear codes, this function measures exactly the dimension of the restricted codes and for a subset $I \subseteq [n]$, $H_{\CC}(I)$ is equivalent to the rank of the sub-matrix formed by the columns $\mathbf{g}_{i}$ with $i \in I$ or, equivalently, to the rank function of $I$ in the associated matroid of $\CC$. The entropy function has the following polymatroidal properties as shown in \cite{fujishige78}.

\begin{proposition}
Let $\CC$ be a quasi-uniform code of length $n$ over the alphabet $Q$ and $H$ the entropy function associated to $\CC$. For $I,J \subseteq [n]$, we have
\begin{enumerate}
\item $H(I) \leq |I|$,
\item If $I \subseteq J$ then $H(I) \leq H(J)$,
\item $H(I \cup J) + H(I \cap J) \leq H(I) + H(J)$.
\end{enumerate}
\end{proposition}

The entropy function also behaves nicely for restricted codes. Let $I \subseteq [n]$ and $\CC|_{I}$ be the restriction of $\CC$ to the set $I$. Then for $J \subseteq I$, we have $H_{\CC|_{I}}(J) = H_{\CC}(J)$. Finally, we define a closure operation on the subsets of $[n]$ for linear codes. 

\begin{definition}
Let $\CC$ be an $[n,k,d]$ linear code and $I \subseteq [n]$. The closure operator $\cl : 2^{[n]} \to 2^{[n]}$ is 
\[
\cl(I)=\{ e \in [n] : H(I \cup e) = H(I) \}.
\]
\end{definition}
One can think of the closure operator via the generator matrix $G_{\CC}$ of $\CC$, where $\cl(I)$ is the set of all columns in $G_{\CC}$ contained in the linear span of the columns indexed by $I$. 

The following table summarizes the notation used throughout the paper. The formal definitions for some of them will only appear later in the paper. 

\begin{center}
{\renewcommand{\arraystretch}{1.2}
\begin{tabular}{|c|c|}
\hline 
$[n,k,d]$ & \makecell{Linear code of length $n$, dimension $k$\\ and minimum distance $d$} \\
\hline 
$(n,k,d,r,\delta)$ & \makecell{LRC with locality $(r,\delta)$ and $\delta$ the local\\ minimum distance} \\
\hline 
$(n,k,d)(\kappa,\delta)$ & LRC with dimension-locality $(\kappa,\delta)$ \\
\hline 
$\GG(k,d)$ & \makecell{Griesmer bound on the length $n$ \\
of a linear code} \\
\hline 
$\kappa_{B}$ & Bound on the dimension of a code \\
\hline 
$\kappa_{A}$ & \makecell{Log-convex bound on the dimension\\ of a code}\\
\hline 
$B_{l-c}(n,d)$ & \makecell{Log-convex bound on the cardinality \\ of a code} \\
\hline 
$\RR$ & Rate $k/n$\\
\hline
$\delta_n$ & Relative minimum distance $d/n$ \\
\hline 
\textit{res}-chain & Chain of consecutive residual codes \\
\hline 
\end{tabular}}
\end{center}

\subsection{Definition of locality and relation with the number of nodes contacted for repairing}

In this part, we explain how the definition of locality $(r,\delta)$ in \cite{prakash12} diverges from the initial motivation of introducing the notion of locality, and state a new definition of locality called dimension-locality. As mentioned in the introduction, \cite{gopalan12} and then \cite{papailiopoulos12} introduced the notion of locality $r$ to reduce the repair traffic by designing storage codes such that one failure can be repaired by contacting only a small number of nodes $r \ll k$ in the storage system. A natural extension of the above definition is to allow multiple erasures to be corrected locally by still accessing  fewer  nodes than $k$. For this, we need the local restrictions to have a minimum distance of $\delta$ so that up to $\delta-1$ erasures can be repaired locally. The definition presented in \cite{prakash12} is the following. 

\begin{definition}
\label{def:loc_rdelta}
An $[n,k,d]$ linear code $\CC$ has \emph{all-symbol locality} $(r,\delta)$ if for all code symbol indices $i \in [n]$ there exists a set $R \subseteq [n]$, called a \emph{repair set}, such that 
\begin{enumerate}
\item $i \in R$,
\item $|R| \leq r + \delta -1$,
\item The minimum distance of the restriction of $\CC$ to the set $R$ is at least $\delta$.  
\end{enumerate}
We refer to $\CC$ as an $(n,k,d,r,\delta)$-LRC. 
\end{definition}

With a slight abuse of language, we say that a repair set $R$ has dimension $\kappa$ and minimum distance $\delta$ if the restricted code $\CC|_{R}$ has dimension $\kappa$ and minimum distance $\delta$. We also say that a repair set $R$ is MDS if $\CC|_{R}$ is an MDS code. 

In Definition \ref{def:loc_rdelta}, any $\delta-1$ coordinates of $R$ are determined by the values of the remaining $|R| - \delta + 1 \leq r $ coordinates, thus enabling local repair by contacting at most $r$ other nodes. The problem with Definition \ref{def:loc_rdelta} is that it implicitly requires the repair sets to be MDS in order for $r$ to be the dimension of the local codes. In other words,  if a repair set $R$ is not MDS, then the number of nodes needed to repair any $\delta-1$ failures in $R$ is strictly less than $r$. Thus, Definition \ref{def:loc_rdelta} diverges from the initial meaning of controlling precisely the number of nodes contacted during the repair process when considering non-MDS repair sets. 

This observation is particularly relevant when the field size is fixed and $\delta$ is too large for MDS codes to exist. For example, if we consider binary codes and require that $\delta>2$ in order to correct more than one erasure locally, then none of the repair sets can be MDS and $r$ is no longer the local dimension. We illustrate this phenomenon by a concrete example. 

\begin{example}
\label{ex:104443}
Let $\CC$ be the binary linear $[10,4,4]$-code given by the  generator matrix
\[ G=\SmallMatrix{
1&0&0&0&1&0&1&1&1&1\\
0&1&0&0&1&1&0&1&1&1\\
0&0&1&0&0&1&0&1&0&1\\
0&0&0&1&0&0&1&0&1&1\\
}.
\]

We define the three repair sets by their corresponding columns in $G$ : 
$
R_{1}=\{1,2,3,5,6,8\}, R_{2}=\{2,3,6,7,9,10\}$, and $R_{3}=\{1,4,6,7,8,10\}$. Every repair set has size $6$, minimum distance $d_{\CC|_{R_{i}}} = 3$, and entropy $H(R_{i})=3$.  Thus, according to Definition~\ref{def:loc_rdelta}, we get $r= |R_i|-\delta+1=4$ and hence $\CC$ is a binary linear $(10,4,4,4,3)$-LRC. However, even though $r=k=4$, we can repair up to two failures by contacting at most $3$ nodes, as shown by the entropy. While this is not really in contrast with the above definition, it only gives a loose bound $r=4$ for the number of nodes we need to contact.
\end{example}

To be able to precisely keep track of the number of nodes contacted during the repair process, we propose a slightly different definition for locally repairable codes tolerating multiple local erasures. We do this by replacing the condition on the size of a repair set by a condition on its dimension. 

\begin{definition}
\label{def:loc_kappadelta}
An $[n,k,d]$ linear code $\CC$ has \emph{all-symbol dimension-locality} $(\kappa,\delta)$ if for all code symbol indices  $i \in [n]$, there exists a set $R \subseteq [n]$ such that 
\begin{enumerate}
\item $i \in R$,
\item $ H(R)= \log_{q}(|\{ \mathbf{c}_{R} : \mathbf{c} \in \CC \}|) \leq \kappa$,
\item The minimum distance of the restriction of $\CC$ to the set $R$ is at least $\delta$.
\end{enumerate}
We refer to $\CC$ as an $(n,k,d)(\kappa,\delta)$-LRC. 
\end{definition}

With this definition, we regain the fact that every $\delta-1$ coordinates can be recovered by contacting at most $\kappa$ other coordinates and $\kappa$ can be made tight. When we do not restrict the field size, optimal repair sets are MDS with size equal to $\kappa + \delta -1$. Thus, both definitions coincide for large alphabets. When the field size is fixed and we have non-MDS repair sets, the parameter $\kappa$ in Definition \ref{def:loc_kappadelta} still measures the local dimension. On the contrary, the parameter $r$ in Definition \ref{def:loc_rdelta} measures the local dimension with the addition of a non-zero part of the local size that varies depending on the field size and the local minimum distance.
Moreover, the new definition allows more flexibility on the size of the repair sets. Specifically, the size of a repair set can be smaller or bigger than $\kappa + \delta -1 $, since $\kappa$ and $\delta$ are only an upper bound on the dimension of a code and a lower bound on its minimum distance, respectively.

Obviously, every code with locality $(r,\delta)$ is a code with dimension-locality $(\kappa = r, \delta)$. To obtain a closer relation between the two notions of locality, we can replace $r$ by the value of an upper bound on the dimension of a code given its length $r+\delta-1$ and minimum distance $\delta$. 
Let $\kappa_{B}$ be the maximal dimension obtained by such a bound. Then, every code with locality $(r,\delta)$ is a code with dimension-locality $(\kappa = \kappa_{B}, \delta)$. The problem is that $\kappa_{B}$ might not be tight, \ie , there is no repair set $R$ such that $H(R)=\kappa_{B}$, which goes against the purpose of the new definition. This is illustrated in the following example.

\begin{example}
\label{ex:199353}
Let $\CC_{1}$ be the $[7,3,4]$ binary linear code obtained by taking the dual of the $[7,4,3]$ Hamming code, and let $\CC_{2}$ be the $[6,3,3]$ code obtained by puncturing $\CC_{1}$ on the last coordinate. Consider the $[13,6,3]$ binary linear code $\CC=\CC_{1} \oplus \CC_{2}$ with generator matrix 
\[ G=\SmallMatrix{
G_{1} & \mathbf{0} \\
\mathbf{0} & G_{2} \\
},
\]
where $G_{1}$ and $G_{2}$ are the generator matrices of $\CC_{1}$ and $\CC_{2}$, respectively. 

By defining the repair sets to be $R_{1}=\{1, \ldots , 7\}$ and $R_{2}=\{8, \ldots , 13 \}$, the code $\CC$ is a $(13,6,3,5,3)$-LRC since $\delta=3$ and $|R_{1}|=7$. The maximal entropy of a repair set is $3$, so $\CC$ is a $(13,6,3)(3,3)$-LRC. To get an upper bound on the local dimension from the parameters $(r,\delta)$, we can use an upper bound on the dimension of a binary code of length $r+\delta-1=7$ and minimum distance $\delta=3$. Since the Hamming code is a $[7,4,3]$ binary code, the best upper bound gives a local dimension of $4$. However, this is strictly greater than the maximal dimension of a local code in this example, which is $3$. Thus, it is impossible to obtain the exact maximal dimension of a local code from the parameters $(r,\delta)$ in Definition \ref{def:loc_rdelta}. 
\end{example}

The previous example also illustrates how the flexibility of the sizes of the repair sets, obtained from Definition \ref{def:loc_kappadelta}, can be used to achieve a code of length $n$ that is not divisible by any of the sizes of the repair sets, while keeping the dimension-locality parameters $(\kappa, \delta)$.




\section{Bounds for dimension-locality $(\kappa,\delta)$ and locality $(r,\delta)$}
\label{section:CMG}

In this section, we study the structure of linear codes with dimension-locality $(\kappa, \delta)$ and derive a bound on their parameters. Following the general framework of \cite{cadambe15}, we construct a set with a large size and a small entropy. This is done by using a detailed analysis of the repair sets based on the work done in \cite{griesmer60} and \cite{solomon65}. It yields a bound of the form of the bound \eqref{eq:CM} handling both MDS and non-MDS repair sets. Then, we extend our bound to linear codes with locality $(r,\delta)$. Finally, using a weaker estimation of our results, we derive a new Singleton-type bound for $(n,k,d,r,\delta)$-LRCs. 

We start by presenting the new bound for linear codes with dimension-locality. Here, $k_{\opt}^{(q)}$ is an upper bound on the dimension of a code and $\GG(\kappa,\delta)$ is the Griesmer bound taken over $\F_{q}$. 

\begin{theorem}
\label{thm:CMG_bound}
Let $\CC$ be a linear $(n,k,d)(\kappa, \delta)$-LRC over $\F_{q}$. Then we have
\begin{equation}
\label{eq:CMG_kappa}
k \leq \min_{\lambda \in \Z_{+}} \left\{ \lambda + k_{\opt}^{(q)}(n-(a+1) \cdot \GG(\kappa, \delta) + \GG(\kappa - b, \delta), d ) \right\}
\end{equation}
where $a,b \in \Z$ are such that $\lambda = a \kappa + b, 0 \leq b < \kappa$.
\end{theorem}

\begin{proof}
The proof is given in the appendix. 
\end{proof}


In order to prove this bound, we need a better understanding of the bound \eqref{eq:CM} and the implications of having a non-MDS repair set. The bound \eqref{eq:CM} relies mainly on two results. The first result is a construction of a set with an upper bound on its entropy and a lower bound on its size. The second result is a shortening argument that governs the part inside $k_{\opt}^{(q)}$ in the bound \eqref{eq:CM}. This is reproduced here with a slight rephrasing.

\begin{lemma}[\cite{cadambe15}, Lemma 2]
\label{lemma:contract}
Let $\CC$ be an $[n,k,d]$ linear code over $\F_{q}$ and $I \subseteq [n]$ such that $H(I) < k $. Then the shortened code $\CC/I$ has parameters $[n-|I|, k-H(I), d'\geq d]$. 
\end{lemma}

Regarding the first result, the technique used to construct large sets relies on taking the union of repair sets. If two repair sets happen to intersect, which will reduce both the entropy and the size of their union, a correction is performed by adding arbitrary elements to their union. The main difficulty to extend this technique to non-MDS repair sets is to deal with the intersection of the repair sets and find the appropriate correction. Specifically, the intersection of two repair sets can now have a size strictly larger than its entropy (take, for example, $R_{1}$ and $R_{2}$ in Example \ref{ex:104443}). Thus, it is not possible anymore to correct their union by an arbitrary set since the resultant set might exceed the upper bound on the entropy. 

In order to correct the intersection, the main idea is to create a set using consecutive residual codes. As mentioned in the introduction, for any $[n,k,d]$ linear code $\CC$ over $\F_q$, there exists $\CC'$, a restriction of $\CC$ called the residual code of $\CC$, such that $\CC'$ has parameters $[n-d, k-1, d'\geq \left\lceil d/q \right\rceil]$. We define the sequence of consecutive residual codes as a chain of subsets of $[n]$.

\begin{definition}
Let $\CC$ be a $[n,k,d]$ linear code over $\F_{q}$. The \emph{res}-chain of $\CC$ is a sequence of sets $(S_{i})_{i=0}^{k}$ with $S_{i} \subseteq [n]$ constructed recursively by starting with $S_{0}=[n]$ and $S_{i}$ is such that $\CC|_{S_{i}}$ is a residual code of $\CC|_{S_{i-1}}$.
\end{definition}
 
This definition is well-defined since by the proof of \cite[Theorem 1']{solomon65}, the residual code $\CC'$ of $\CC$ is constructed by restricting $\CC$ to a well-chosen set of coordinates. Therefore, we can interpret the recursive residual code chain as a sequence of sets in $[n]$. Furthermore, as the dimension of the residual code is one less than the dimension of the original code, the chain has length $k+1$ and for all $0 \leq \alpha \leq k$, there is a set $S$ in the \textit{res}-chain of $\CC$ such that $H(S)=\alpha$. Finally, by a recursive argument, if $S$ is a set in the \textit{res}-chain of $\CC$, then the minimum distance $d_{S}$ of the restriction to $S$ is bounded from below by
\[
d_{S} \geq \left\lceil \frac{d}{q^{k-H(S)}} \right\rceil.
\]


We now present two lemmas that are used to prove Theorem \ref{thm:CMG_bound}. The first lemma states how to increase the size of a set when the addition of an entire repair set will exceed the desired entropy. 

\begin{lemma}
\label{lemma:correction}
Let $\CC$ be a linear $(n,k,d)(\kappa,\delta)$-LRC over $\F_{q}$. Let $F\subseteq [n]$ be such that $\cl(F)=F$ and $\alpha$ an integer with $1 \leq \alpha \leq \kappa$. If there exists a repair set $R$ such that $H(R) - H(F \cap R) \geq \alpha$, then, there exists $F' \subseteq [n]$ with $\cl(F')=F'$ such that
\begin{itemize}
\item $H(F') \leq H(F) + \alpha$,
\item $|F'| \geq |F| + \GG \left( \alpha, \left\lceil \frac{\delta}{q^{\kappa - \alpha}} \right\rceil \right)$.
\end{itemize}
\end{lemma}

\begin{proof}
The proof is given in the appendix.
\end{proof}

Using the above lemma, we can prove the following second lemma that represents the challenging part of proving the new bound. 

\begin{lemma}
\label{lemma:recursive_correction}
Let $\CC$ be a linear $(n,k,d)(\kappa, \delta)$-LRC over $\F_{q}$. Let $F \subseteq [n]$ be such that $\cl(F)=F$ and $H(F) +\kappa \leq k$. Then, there exists $F_{c} \subseteq [n]$ with $\cl(F_{c})=F_{c}$ such that 
\begin{itemize}
\item $H(F_{c}) \leq H(F) + \kappa$,
\item $|F_{c}| \geq |F| + \GG(\kappa, \delta)$.
\end{itemize}
\end{lemma}

\begin{proof}
The proof is given in the appendix.
\end{proof}

The intuition behind the proofs of Lemma \ref{lemma:correction} and  \ref{lemma:recursive_correction} is the following. If all the repair sets are disjoint and have dimension $\kappa$, then Lemma \ref{lemma:recursive_correction} follows directly since no correction is needed. If two repair sets intersect each other or have a dimension less than $\kappa$, both the entropy and the size of their union will be smaller than expected and a correction is required. In this case, we can use the chain of residual codes to get a set that, when added to their union, increases the entropy by exactly the amount left. The last trick is to evaluate both the size of the union and the set in the \textit{res}-chain using the Griesmer bound. First, it is a bound on the length of a code where the minimum distance plays a more important role compared to the dimension. This fits the lower bound on the local minimum distance for codes with dimension-locality. Secondly, the Griesmer bound has the nice property that $\GG(\kappa, \delta) = \GG(\alpha, \delta) + \GG \left(\kappa - \alpha, \left\lceil\frac{\delta}{q^{\alpha}} \right\rceil \right) $. The first term in the sum can be used to get a lower bound on the size of the repair set minus its intersection with a given set. The second term, under some conditions, gives a lower bound on the size of a particular set in the \textit{res}-chain of a repair set. Thus, this relation is really useful when we add the extra set to correct the union of a repair set to a given set. Finally, the Griesmer bound is also consistent with our construction based on residual codes. 

As a corollary of Theorem \ref{thm:CMG_bound}, we can force the parameter $\lambda$ to only be a multiple of $\kappa$. This gives a bound resembling the original bound in \cite{cadambe15}. 

\begin{corollary}
\label{cor:CMG_tkappa}
Let $\CC$ be a linear $(n,k,d)(\kappa, \delta)$-LRC over $\F_{q}$. Then we have
\begin{equation}
\label{eq:CMG_tkappa}
k \leq \min\limits_{t \in \Z_{+}} \left\{ t\kappa + k_{\opt}^{(q)}(n-t\cdot \GG(\kappa, \delta), d) \right\}.
\end{equation}
\end{corollary}

Even if the wider range of the parameter $\lambda$ makes the bound \eqref{eq:CMG_kappa} theoretically better than the bound \eqref{eq:CMG_tkappa}, the two bounds show similar experimental results. More precisely, we randomly generated some feasible parameters $(n,k,d,r,\delta)$ for LRCs over the binary field, which yielded results showing that the bound \eqref{eq:CMG_tkappa} is equal to the bound \eqref{eq:CMG_kappa}. One possible justification is that for two consecutive dimensions $\lambda-1$ and $\lambda$ inside the minimum in \eqref{eq:CMG_kappa}, the length in the second term decreases by the largest value when $\lambda = t \kappa$. Therefore, the optimal condition on the global dimension $k$ would always happen when $\lambda$ is a multiple of the local dimension $\kappa$. However, a formal proof is impossible due to the unknown intrinsic bound $k_{\opt}^{(q)}$.

\subsection{New bounds for locality $(r,\delta)$}

As already explained in Section \ref{section:definitions}, to obtain a bound on LRCs with locality $(r,\delta)$ instead of dimension-locality $(\kappa, \delta)$, we can estimate $\kappa$ by an upper bound on the maximal dimension of a code given its length $r + \delta -1$ and its minimum distance $\delta$. Let us call $\kappa_{B}$ the value of an arbitrary upper bound on the maximal dimension of a repair set. Since we never used that $\kappa$ is actually tight, our previous results apply directly to codes with locality $(r,\delta)$ via the estimated dimension $\kappa_{B}$. Therefore, we get the following new bound.

\newcommand{\seq}{{\leq}\!}
\begin{theorem}
\label{thm:CMG_rdelta}
Let $\CC$ be a linear $(n,k,d,r,\delta)$-LRC over $\F_{q}$ and $\kappa_{B}$ the upper bound on the local dimension. Then
\begin{equation}
\label{eq:CMG_r}
k \seq \min_{\lambda \in \Z_{+}} \left\{ \hspace{-1pt}\lambda + k_{\opt}^{(q)}(n-(a+1) \GG(\kappa_{B}, \delta) + \GG(\kappa_{B} - b, \delta), d ) \hspace{-2.5pt} \right\},
\end{equation}
where $a,b \in \Z$ are such that $\lambda = a \kappa_{B} + b, 0 \leq b < \kappa_{B}$.
\end{theorem}

It is crucial to estimate the length in the shortened part of the bound via the Griesmer bound instead of replacing it by $r+\delta-1$. The reason is that $r+\delta-1$ is an upper bound on the size of a repair set while we need something of the form of a lower bound. However, what we need is not exactly a lower bound since the dimension of a repair set can be lower than $\kappa_{B}$. We present a counter-example. 

\begin{example}
\label{ex:103323}
Let $\CC$ be the binary linear code given by the following generator matrix
\[ G=\SmallMatrix{
1&1&1&1&0&0&0&0&0&0\\
0&0&0&0&1&1&1&0&0&0\\
0&0&0&0&0&0&0&1&1&1\\
}.
\]
$\CC$ is a $(10,3,3,2,3)$-LRC with obvious repair sets. Estimating $\kappa_{B}$ using the Griesmer bound yields $\kappa_{B}=1$. However, there is no sets with an entropy less than $2$ and a size greater than $2\cdot (r+\delta-1)= 8$ since every set of size $8$ has already an entropy equal to $3$. The problem here is that the repair set of size $4$, which gives the upper bound $r+\delta-1$, has a minimum distance strictly greater than $\delta=3$. 
\end{example}

Using Lemma \ref{lemma:recursive_correction}, we can derive a Singleton-type bound that take into consideration non-MDS repair sets. 

\begin{theorem}
\label{thm:Singleton_G}
Let $\CC$ be a linear $(n,k,d,r,\delta)$-LRC over $\F_{q}$ and $\kappa_{B}$ the upper bound on the local dimension. Then 
\begin{equation}
\label{eq:singl_G}
d \leq n - \left\lceil \frac{k}{\kappa_{B}} \right\rceil \GG(\kappa_{B}, \delta) + \GG(\kappa_{B} - b, \delta)
\end{equation}
where $b=k-1 - \left( \left\lceil \frac{k}{\kappa_{B}} \right\rceil -1 \right) \kappa_{B}$. 
\end{theorem}

\begin{proof}
Let $a,b \in \Z$ be such that $k-1 = a\kappa_{B} + b$ with $0 \leq b < \kappa_{B}$. By Lemma \ref{lemma:recursive_correction} and the proof of Theorem \ref{thm:CMG_bound}, there is a set $I \subseteq [n]$ such that $H(I) \leq k-1$ and $|I| \geq (a+1) \GG(\kappa_{B}, \delta) + \GG(\kappa_{B}-b, \delta)$. Then, the minimum distance $d$ is bounded by
\[
d \leq n- |I| \leq n - \left\lceil \frac{k}{\kappa_{B}} \right\rceil \GG(\kappa_{B}, \delta) + \GG(\kappa_{B} - b, \delta).
\]
\end{proof}



\section{Analysis and comparisons}
\label{section:analysis}

This section is devoted to the comparison between our bounds and the previously known bounds for LRCs. In the first part, we show that the bound \eqref{eq:CMG_r} leads to the straightforward extension of the bound \eqref{eq:CM} for locality $(r,\delta)$ and the bound \eqref{eq:singl_G} leads to the Singleton-type bound \eqref{eq:Prakash} when the field size is sufficiently large. In the second part, we derive the asymptotic formulas of the bounds \eqref{eq:CMG_r} and \eqref{eq:singl_G} to obtain the bounds on the tradeoff between the rate and the relative minimum distance of LRCs for fixed locality. This also enables a cleaner comparison between the new bounds and the bound \eqref{eq:ABHMT} from \cite{agarwal18}. Notice that we do not compare our bounds to the linear programming bound derived in \cite{agarwal18} since it is impossible to derive an asymptotic formula from it and we do not assume that the repair sets are disjoint. 

Our results show that the comparison between the new asymptotic Singleton-type bound and the asymptotic version of bound \eqref{eq:ABHMT} depends on the performance of the Griesmer bound compared to the log-convex bounds. Specifically, we give some examples where our bound is better, equal, or worse than the bound \eqref{eq:ABHMT}. On the other hand, we prove that the bound \eqref{eq:CMG_r} is always better than the bound \eqref{eq:ABHMT} for large relative minimum distances by using the Plotkin bound as the intrinsic bound in \eqref{eq:CMG_r}. 

We start by showing that the bound \eqref{eq:CMG_r} leads to the straightforward extension of the bound \eqref{eq:CM} for locality $(r,\delta)$.

\begin{corollary}
Let $\CC$ be a linear $(n,k,d,r,\delta)$-LRC over $\F_{q}$. Then
\begin{equation}
\label{eq:CM_rdelta}
k \leq \min\limits_{t \in \Z_{+}} \left\{ tr + k_{\opt}^{(q)}(n-t(r+\delta-1), d) \right\}.
\end{equation}
\end{corollary}

This extension was presented in \cite[Remark 3]{rawat15} for codes with $(r',l)$-cooperative locality. It can be adapted to LRCs by using the fact that a code with locality $(r,\delta)$ is a code with $(r'=r,l=\delta-1)$-cooperative locality. 

\begin{proof}
Let $\kappa_{B}$ be the upper bound on the local dimension. We want to show that for all $t \in \Z_{+}$ with $t \leq \frac{k}{r}$, there is a set $I$ with $H(I)\leq tr$ and $|I| \geq t(r+\delta-1)$. For $t$ fixed, define $\lambda, a , b \in \Z_{+}$ such that $\lambda = tr = a\kappa_{B} + b$. By the same arguments as in the proof of Theorem \ref{thm:CMG_bound}, there exists a set $I$ such that $H(I) \leq \lambda =tr$. It remains to show that $|I| \geq t(r+\delta-1)$. First, we have $a \geq t$ since $a=\frac{tr-b}{\kappa_{B}} \geq \frac{t \kappa_{B} -b}{\kappa_{B}} = t - \frac{b}{\kappa_{B}}$. Now $a$ is an integer, so $a \geq \left\lceil t - \frac{b}{\kappa_{B}} \right\rceil = t$. Using the fact that the Griesmer bound is greater than or equal to the Singleton bound, \ie , $\GG(\kappa_B, \delta) \geq \kappa_B + \delta -1$,  we have
\begin{align*}
|I| & \geq a\GG(\kappa_B, \delta) + \GG(\kappa_B, \delta) - \GG(\kappa_B - b, \delta) \\
& \geq a(\kappa_B + \delta -1) + b \\
& = tr + a(\delta-1) \\
& \geq t(r+\delta-1).
\end{align*}
Hence, using Lemma \ref{lemma:contract} with this approximation on the size of $I$, we obtain the desired bound on $k$. 
\end{proof}

Now, we prove that the new Singleton-type bound can be used to obtain the bound \eqref{eq:Prakash}. 

\begin{proposition}
For any linear $(n,k,d,r,\delta)$-LRC, the bound \eqref{eq:singl_G} is at least as strong as the bound~\eqref{eq:Prakash}. 
\end{proposition}

\begin{proof}
We rewrite the bound of Theorem \ref{thm:Singleton_G} to have something closer to the form of the bound \eqref{eq:Prakash}. First, we rewrite the Griesmer bound as
\begin{align*}
\GG(\kappa_{B}, \delta) & = \sum\limits_{i=0}^{\kappa_{B} - 1}\left\lceil \frac{\delta}{q^{i}} \right\rceil  = \kappa_{B} + \sum\limits_{i=0}^{\kappa_{B} - 1}\left( \left\lceil \frac{\delta}{q^{i}} \right\rceil -1 \right) \\
& = \kappa_{B} + \sum\limits_{i=0}^{\kappa_{B} - 1}\left\lfloor \frac{\delta-1}{q^{i}} \right\rfloor.
\end{align*}
Let $a,b \in \Z$ be such that $k-1=a\kappa_{B} + b, 0 \leq b < \kappa_{B}$. The bound of Theorem \ref{thm:Singleton_G} can be transformed as follows:
\begin{small}
\begin{align*}
d & \leq n - a\GG(\kappa_{B}, \delta) - \GG(\kappa_{B}, \delta) + \GG(\kappa_{B} - b, \delta) \\
& = n - a \left( \kappa_{B} + \sum\limits_{i=0}^{\kappa_{B} -1}\left\lfloor \frac{\delta-1}{q^{i}} \right\rfloor \right) - \left( b + \sum\limits_{i=\kappa_{B} - b}^{\kappa_{B} -1}\left\lfloor \frac{\delta-1}{q^{i}} \right\rfloor  \right) \\
& = n - k + 1 - \left( \left\lceil \frac{k}{\kappa_{B}} \right\rceil -1 \right) \sum\limits_{i=0}^{\kappa_{B} -1}\left\lfloor \frac{\delta-1}{q^{i}} \right\rfloor \\
& \hspace{20pt}- \sum\limits_{i=\kappa_{B} - b}^{\kappa_{B} -1}\left\lfloor \frac{\delta-1}{q^{i}} \right\rfloor. \\
\end{align*}
\end{small}
By using the fact that $\kappa_{B} \leq r \Rightarrow  \left\lceil \frac{k}{\kappa_{B}} \right\rceil \geq  \left\lceil \frac{k}{r} \right\rceil$ and $\sum\limits_{i=0}^{\kappa_{B} -1}\left\lfloor \frac{\delta-1}{q^{i}} \right\rfloor \geq \delta-1$, we obtain
\begin{align*}
n - & k + 1 - \left( \left\lceil \frac{k}{\kappa_{B}} \right\rceil -1 \right) \sum\limits_{i=0}^{\kappa_{B} -1}\left\lfloor \frac{\delta-1}{q^{i}} \right\rfloor - \sum\limits_{i=\kappa_{B} - b}^{\kappa_{B} -1}\left\lfloor \frac{\delta-1}{q^{i}} \right\rfloor \\
& \leq n-k+1 - \left( \left\lceil \frac{k}{r} \right\rceil-1 \right)(\delta-1).
\end{align*} 
\end{proof}

This shows that the bounds \eqref{eq:CMG_r} and \eqref{eq:singl_G} are at least as good as the bounds \eqref{eq:CM_rdelta} and \eqref{eq:Prakash}, respectively. Furthermore, we can see that the bounds \eqref{eq:CMG_r} and \eqref{eq:singl_G} improve on the previous bounds when $\kappa_{B} < r$ or when $\delta>q$. The latter case is of particular interest for small alphabets. For example, when considering binary LRCs, the new bounds are already better than the bound \eqref{eq:Prakash} for all $\delta\geq 3$. 

\subsection{Asymptotic regime}

For the rest of this section, we look at the asymptotic regime where $n \to \infty$. Let $\RR=k/n$ be the rate of a code and $\delta_{n} = d/n$ its relative minimum distance. The goal is to obtain the bounds on the tradeoff between the rate and the relative minimum distance of LRCs when the locality $(r,\delta)$ is fixed and $n \to \infty$. This also makes the comparison to the bound \eqref{eq:ABHMT} easier. 

We start with the Singleton-type bound \eqref{eq:singl_G}. By dividing the bound \eqref{eq:singl_G} by $n$ and letting $n \to \infty$, its asymptotic formula is as follows :
\begin{equation}
\label{eq:Single_G_asympt}
\RR \leq \frac{\kappa_{B}}{\GG(\kappa_{B},\delta)}(1-\delta_{n}) + o(1).
\end{equation}

Following the same method, we can derive the asymptotic version of the bound \eqref{eq:ABHMT}. For the ease of reading, we reproduce here the bound  : For any $(n,k,d,r,\delta)$-LRC over $\F_{q}$ and $B_{l-c}(n,d)$ a bound on the cardinality of a code, which is log-convex in $n$ and such that $B_{l-c}(0,d)=1$, we have
\[
k \leq \left( \left\lceil \frac{n-d+1}{r+\delta-1} \right\rceil +1 \right) \log_{q} B_{l-c}(r+\delta-1,\delta).
\]

Its asymptotic version is therefore : 
\begin{equation}
\label{eq:ABHMT_asympt}
\RR \leq \frac{\log_{q} B_{l-c}(r+\delta-1,\delta)}{r + \delta -1} (1-\delta_{n}) + o(1).
\end{equation}

The following table summarizes the asymptotic formulas for the Singleton-type bounds with different locality assumptions. Notice that the last three are truly comparable since they share the same locality assumptions. When looking at the table, we can see how the locality assumption reduces the rate by the fraction of the local dimension over the local size.

\begin{center}
{\renewcommand{\arraystretch}{1.2}
\begin{tabular}{|c|c|}
\hline 
Singleton bound & $\RR \leq 1 - \delta_{n} + o(1)$ \\ 
\hline 
\makecell{Gopalan et al.\\ \cite{gopalan12}} & $\RR \leq \frac{r}{r+1}(1-\delta_{n}) + o(1)$ \\ 
\hline 
\makecell{Prakash et al.\\ \cite{prakash12}} & $\RR \leq \frac{r}{r+\delta -1}(1-\delta_{n}) + o(1)$ \\ 
\hline 
\makecell{Agarwal et al.\\ \cite{agarwal18}} & $\RR \leq \frac{\log_{q} B_{l-c}(r+\delta-1,\delta)}{r+\delta-1}(1-\delta_{n}) + o(1)$ \\
\hline
Theorem \ref{thm:Singleton_G} & $\RR \leq \frac{\kappa_{B}}{\GG(\kappa_{B}, \delta)}(1-\delta_{n}) + o(1)$ \\ 
\hline 
\end{tabular}}
\end{center}

Following the method in \cite{cadambe15}, we can derive the asymptotic formula for the bound \eqref{eq:CMG_r}. Define $R_{\opt}(\delta_{n}) = \lim\limits_{n \to \infty} \frac{k_{opt}^{(q)}(n,\delta_{n}n)}{n}$. By dividing the bound \eqref{eq:CMG_r} by $n$ , we obtain its asymptotic version 
\begin{equation}
\label{eq:CMG_asympt}
\RR \leq \min\limits_{0 \leq x < \frac{1}{\nu}} x + \left(1-x\nu \right) R_{\opt}\left( \frac{\delta_{n}}{1-x\nu} \right),
\end{equation}
where $\nu = \GG(\kappa_{B},\delta)/\kappa_{B}$.

We can now compare the asymptotic formulas between \eqref{eq:Single_G_asympt}, \eqref{eq:CMG_asympt}, and \eqref{eq:ABHMT_asympt}. Notice that for linear codes,  $\log_{q} B_{l-c}(r+\delta-1,\delta)$ is a bound on the dimension of a repair set. From now on, we denote by $\kappa_{A}$ the bound $\kappa_{A}:=\log_{q} B_{l-c}(r+\delta-1,\delta)$. By definition, $\kappa_{A}$ gives a valid upper bound on the local dimension in Theorem \ref{thm:CMG_rdelta}. However, the best upper bound on the dimension of a code is not necessarily log-convex and hence, $\kappa_{B}\leq \kappa_{A}$. In particular, the Griesmer bound on the cardinality of a code is not a log-convex function on the length $n$ as demonstrated next.

Remember that a positive function $f(j)$ of the integer argument is called log-convex if $f(j_{1})f(j_{2}) \leq f(j_{1}-1)f(j_{2}+1)$ for any $j_{1}\leq j_{2}$ in the support of $f$. For any $[n,k,d]$ linear code $\CC$ over $\F_{q}$, the Griesmer bound on $k$ given $n$ and $d$, denoted by $\GG_{k}(n,d)$, is obtained by taking the maximal $k' \in \Z_{+}$ such that $\GG(k',d) \leq n$. Thus, it gives a bound on the cardinality, $|\CC| \leq q^{k'}$. Let us consider the parameters $n_{1}=n_{2}=8$, $d=5$, and $q=2$. Then, we obtain
\begin{align*}
&\GG_{k_{1}}(8,5)=2, \\
&\GG_{k'_{1}}(7,5)=1, \\
&\GG_{k'_{2}}(9,5)=2.
\end{align*}
Hence, we have $2^{k_{1}}2^{k_{2}} = 2^{4} > 2^{k'_{1}}2^{k'_{2}}=2^{3}$ and the Griesmer bound on the cardinality of a code is not log-convex on $n$. Therefore, there is no obvious answer to the comparison between the bounds \eqref{eq:Single_G_asympt} and \eqref{eq:ABHMT_asympt} since we need to compare $\frac{\kappa_{B}}{\GG(\kappa_{B},\delta)}$ and $\frac{\kappa_{A}}{r+\delta-1}$, and both the numerator and the denominator of the former are smaller than or equal to their respective correspondents in the latter. 

To be more specific, the comparison between the two bounds \eqref{eq:Single_G_asympt} and \eqref{eq:ABHMT_asympt} mainly depends on the performance of the Griesmer bound compared to the log-convex bounds. For example, if there exists a log-convex bound such that $\kappa_{B}=\kappa_{A}$ but $\GG(\kappa_{B},\delta)<r+\delta-1$, then the bound \eqref{eq:ABHMT_asympt} is strictly better than the new Singleton-type bound. This is illustrated in Figure \ref{fig:ABHMT_better}, which displays the rate--distance tradeoff for binary codes with locality $(6,3)$. To evaluate the local dimension, we use the Hamming bound as a log-convex bound to get $\kappa_{A}=4$, which is optimal. The Griesmer bound gives $\GG(4,3)=7<8$. Hence the green line representing the bound \eqref{eq:ABHMT_asympt} is better than the orange line displaying the bound \eqref{eq:Single_G_asympt}. 

\begin{figure}
\centering
\includegraphics[height=6.5cm]{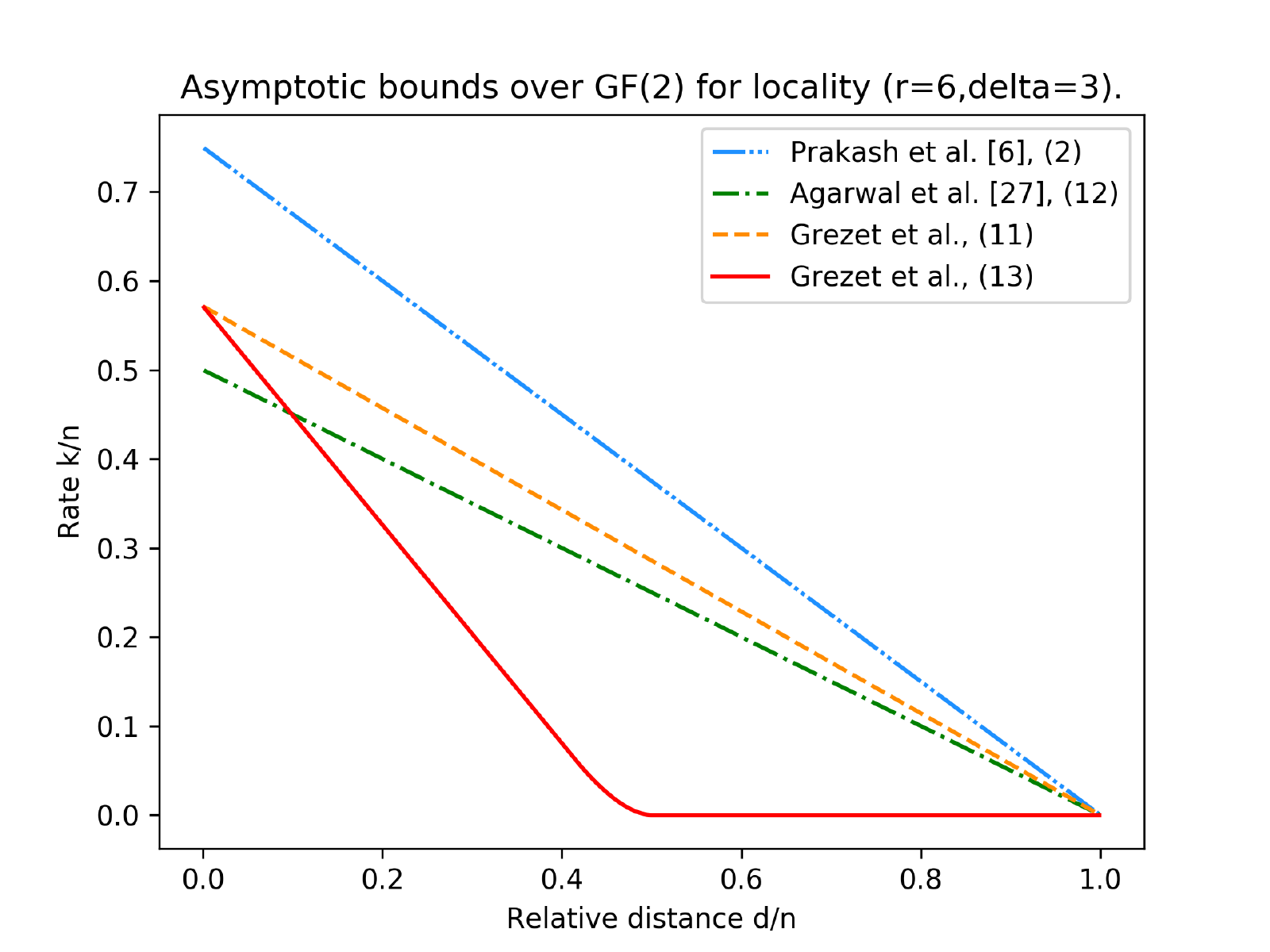}
\caption{Comparison of the asymptotic upper bounds on the rate $k/n$ from \cite{prakash12, agarwal18}, and the bounds \eqref{eq:Single_G_asympt} and  \eqref{eq:CMG_asympt} over the binary field with fixed locality $(r=6,\delta=3)$.}
\label{fig:ABHMT_better}
\end{figure}

On the other hand, if $\GG_{k}(r+\delta-1,\delta)< \kappa_{A}$ for all log-convex bounds then the bound \eqref{eq:Single_G_asympt} is strictly better than the bound \eqref{eq:ABHMT_asympt} because we have $\GG(\kappa_{A},\delta)>r+\delta-1$. Since it is impossible to give a proper example due to the fact that we would need to prove it for all log-convex bounds, we restrict here the comparison between the two bounds by considering the three bounds proven to be log-convex in \cite{agarwal18}, namely the Singleton, Hamming and Plotkin bounds. Let $\CC$ be a linear LRC with locality $(12,9)$. The Singleton bound gives an upper bound on the local dimension of 12 and the Hamming bound gives a bound of 7. The Plotking bound is not applicable here since $\delta<(r+\delta-1)/2$. Now, the Griesmer bound on the dimension of a code gives an upper bound of $5$ and is therefore better than the Hamming bound. As we can see in Figure \ref{fig:Singl_G_better} displaying the asymptotic bounds for binary codes with locality $(12,9)$, the bound \eqref{eq:Single_G_asympt} in orange is always better than the bound \eqref{eq:ABHMT_asympt} in green.

\begin{figure}
\centering
\includegraphics[height=6.5cm]{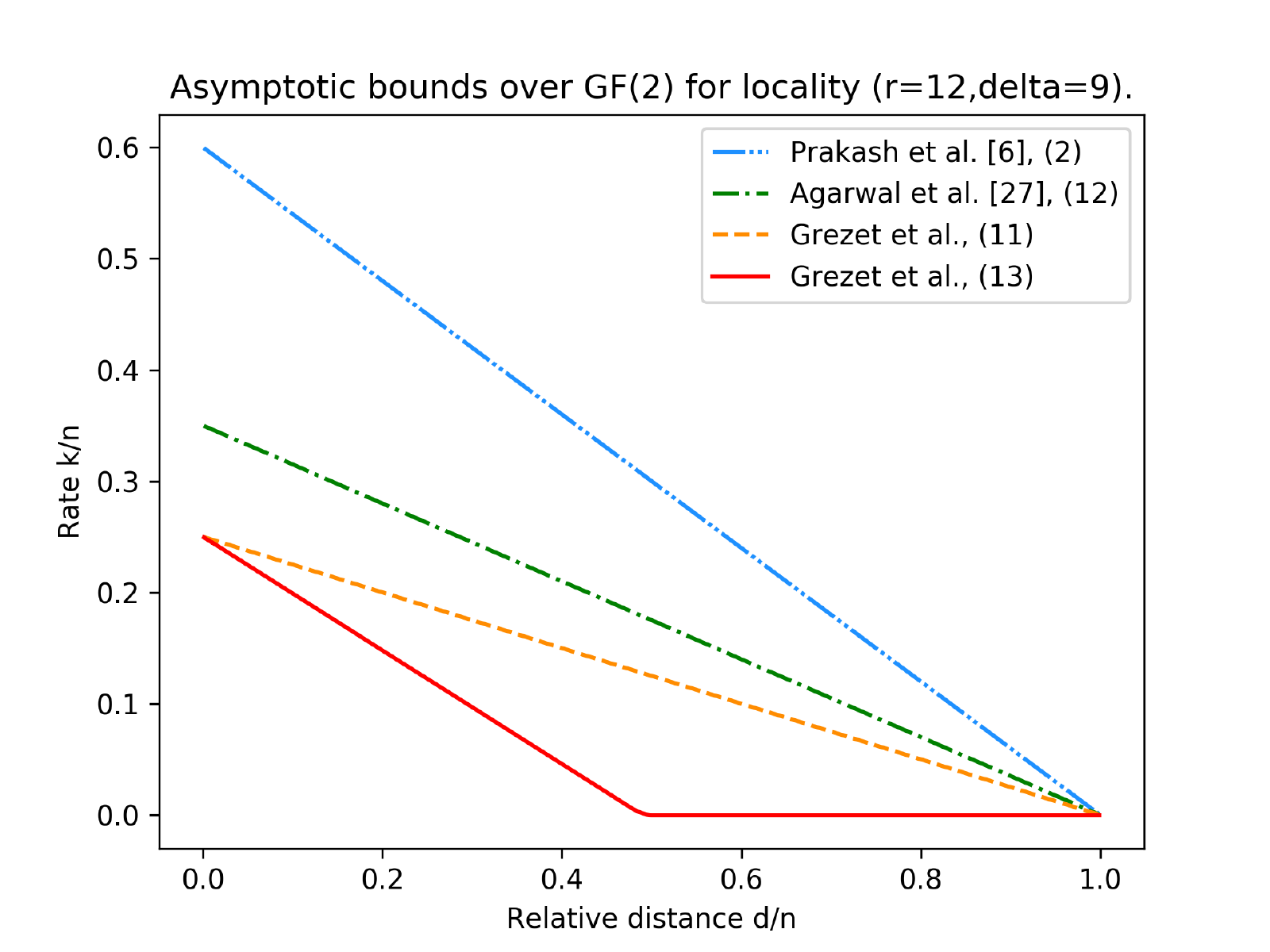}
\caption{Comparison of the asymptotic upper bounds on the rate $k/n$ from \cite{prakash12, agarwal18}, and the bounds \eqref{eq:Single_G_asympt} and  \eqref{eq:CMG_asympt} over the binary field with fixed locality $(r=12,\delta=9)$.}
\label{fig:Singl_G_better}
\end{figure}

Finally, if $\kappa_{A}=\kappa_{B}$ and $\GG(\kappa_{A},\delta)=r+\delta-1$, the two bounds become the same. This happens, for example, in Figure \ref{fig:asympt_comp} where the locality is $(4,3)$ and both the Griesmer bound and the Plotkin bound give $\kappa_{A}=3$ and $\GG(3,3)=6=r+\delta-1$. 

Nonetheless, these are just special cases of the comparison between the bounds \eqref{eq:Single_G_asympt} and \eqref{eq:ABHMT_asympt} and the final comparison needs to be done case-specific. In particular, it depends on three parameters impossible to compute theoretically, which are the best upper bound on the dimension of a code $\kappa_{B}$, the best log-convex upper bound on the dimension of a code $\kappa_{A}$, and the performance of the Griesmer bound regarding both $\kappa_{A}$ and $\kappa_{B}$. 

The comparison between the bound \eqref{eq:ABHMT_asympt} and the bound \eqref{eq:CMG_asympt} is more straightforward. Since the latter is at least as good as the bound \eqref{eq:Single_G_asympt}, we automatically get that the bound \eqref{eq:CMG_asympt} is stronger than the bound \eqref{eq:ABHMT_asympt} when the new Singleton-type is stronger than or equal to the bound \eqref{eq:ABHMT_asympt}. Furthermore, we will see that the bound \eqref{eq:CMG_asympt} is always stronger than the bound \eqref{eq:ABHMT_asympt} for large relative minimum distances, \ie , there is a threshold value $\delta_{t}$ such that for all relative minimum distances $\delta_{n} \in [\delta_{t},1]$, the bound \eqref{eq:CMG_asympt} is  better than the bound \eqref{eq:ABHMT_asympt}.

To prove this, we use the asymptotic Plotkin bound \cite[Theorem 5.2.5]{vanLint99} for $R_{\opt}$ given by

\[
R_{\opt}(\delta_{n}) \leq 1 - \frac{q}{q-1}\delta_{n} + o(1).
\]

Combining it with the bound \eqref{eq:CMG_asympt} and solving the optimization problem yields the following bound on the rate
\begin{equation}
\label{eq:CMG_Pltk_asympt}
\RR \leq \frac{\kappa_{B}}{\GG(\kappa_{B},\delta)} \left(1-\frac{\delta_{n}}{1-1/q} \right) + o(1).
\end{equation}

We can now state formally our claim.

\begin{proposition}
\label{prop:asympt_threshold}
Let $\CC$ be a linear $(n,k,d,r,\delta)$-LRC and $\kappa_{A}$ be the bound on the local dimension given by the best log-convex bound $B_{l-c}(n,d)$. Assume that $\GG(\kappa_{A},\delta) < r+\delta-1$ and let 
\[
\delta_{t} := \frac{1}{1+\frac{1}{q-1} \left( \dfrac{1}{1-\frac{\GG(\kappa_{A},\delta)}{r+\delta-1}} \right)}.
\]
Then the bound \eqref{eq:CMG_Pltk_asympt} is stronger than the bound \eqref{eq:ABHMT_asympt} for all relative minimum distances $\delta_{n} \in (\delta_{t},1]$. 
\end{proposition}

\begin{proof}
The proof is given in the appendix.
\end{proof}

The proof follows from the fact that the bound \eqref{eq:CMG_Pltk_asympt} using $\kappa_{A}$ and the bound \eqref{eq:ABHMT_asympt} are two lines with different slopes and that the bound \eqref{eq:CMG_Pltk_asympt} is equal to $0$ when $\delta_{n}$ is larger than $\frac{q-1}{q}$. Thus, the two lines intersect exactly in $\delta_{t}$ and the bound \eqref{eq:CMG_Pltk_asympt} is better than the bound \eqref{eq:ABHMT_asympt} for relative minimum distances strictly greater than $\delta_{t}$. 

Finally, any bound on the rate improving on the asymptotic Plotkin bound will thus increase the size of the interval where the bound \eqref{eq:CMG_Pltk_asympt} is better than the bound \eqref{eq:ABHMT_asympt}. In particular, this is true for the rate--distance bound given in \cite{McEliece77}, which is the best known bound for binary code. The MRRW bound is as follows:
\begin{equation}
\label{eq:MRRW}
\RR(\delta_{n}) \leq h(1/2 - \sqrt{\delta_{n}(1-\delta_{n})} + o(1)),
\end{equation}
where $h(x):=-x\log_{2}x - (1-x)\log_{2}(1-x)$ is the binary entropy function.

We can thus replace the asymptotic Plotkin bound with the MRRW bound in \eqref{eq:CMG_asympt}. By numerically solving the optimization problem, we obtain the red curve in Figures \ref{fig:asympt_comp}, \ref{fig:ABHMT_better}, and \ref{fig:Singl_G_better}. We see that the bound \eqref{eq:CMG_asympt} combining with the MRRW bound improves significantly on the bound \eqref{eq:ABHMT_asympt} even when the maximal size of a repair set is larger than the Griesmer bound.


\section{Achievability results}
\label{section:simplex}

Several constructions of LRCs achieving the bound \eqref{eq:Prakash} already exist, for example in \cite{prakash12, kamath13b, rawat14b, ernvall16, westerback16}. Many of these constructions require an alphabet size to be exponential in the code length. A construction of LRCs with locality $(r,\delta)$ achieving the bound \eqref{eq:Prakash} and with linear field size was presented in \cite{tamo14}. Since the bound \eqref{eq:CMG_r} approaches the bound \eqref{eq:Prakash} for large alphabets, the bound \eqref{eq:CMG_r} is tight. In this section, we show that the bound \eqref{eq:CMG_r} is also tight for some parameter values for every fixed field size, in particular small ones, by considering the family of the Simplex codes. 

\begin{definition}
Let $n=\frac{q^{m}-1}{q-1}$ and $G_{m,q}(\CC)$ be an $m\times n$ matrix over $\F_{q}$ with non-zero pairwise independent columns. The code $\CC$ generated by $G_{m,q}(\CC)$ is called a $q$-ary Simplex code $S(m,q)$ with parameters $[(q^{m}-1)/(q-1), m , q^{m-1} ]$. 
\end{definition}

Since the Simplex codes are known to achieve the Griesmer bound, they will achieve the bound \eqref{eq:CMG_r} by taking $\lambda=0$ and using the Griesmer bound for $k_{\opt}$. Therefore, the locality parameters do not influence the optimality of the code. This is in fact true in general. If a code already achieves a bound on $[n,k,d]$ without locality constraints and has a certain locality, then it will be an optimal locally repairable code for these locality parameters by the bound \eqref{eq:CMG_r}.

The Simplex codes, as locally repairable codes with $\delta=2$, were already considered in \cite{cadambe15}. In \cite{silberstein18}, the authors used them to construct new LRCs. Here, we want to derive the locality for larger dimensions and $\delta>2$. 

The first thing to notice is that, for every coordinate $e \in [n]$, there exists a codeword $c$ in $S(m,q)$, different from the zero codeword, such that $c_{e}=0$. Indeed, it is enough to take two different codewords not a multiple of each other and subtract them in an appropriate manner to obtain the desired codeword. Since every codeword of $S(m,q)$ has the same weight, we can take the residual code associated to $c$, which is the Simplex code $S(m-1,q)$. By recursion, for every $1 \leq \kappa \leq m$, the coordinate $e$ is contained in a Simplex code $S(\kappa,q)$ obtained by a restriction of $S(m,q)$. Since the minimum distance of $S(\kappa,q)$ is $q^{\kappa-1}$, by letting $\kappa\geq 2$, we ensure that $\delta>1$. Hence, the Simplex code $S(m,q)$ has dimension-locality $\left( \kappa, q^{\kappa-1} \right)$ for all $2 \leq \kappa \leq m$. Finally, $r=\frac{q^{\kappa}-1}{q-1} - q^{\kappa-1} +1 = \frac{q^{\kappa-1}+q -2}{q-1}$. 

To get examples that achieve the bounds \eqref{eq:CMG_r} and \eqref{eq:singl_G} in a less obvious manner, we prove that all the examples presented in this paper are optimal. Let us start with Example \ref{ex:104443}, where the code $\CC$ is a binary $(10,4,4,4,3)$-LRC with $\kappa=3$. By using the Plotkin bound, we get $\kappa_{B}=\kappa=3$ and thus $\GG(3,3)=6$. We can now compute the bound \eqref{eq:singl_G}:
\[
d \leq 10 - \left\lceil \frac{4}{3} \right\rceil \cdot 6 + 6 = 4.
\]
Hence, $\CC$ is optimal. 

In Example \ref{ex:199353}, we presented a binary code with parameters $(13,6,3,5,3)$ and $\kappa=3$. Since the purpose of this example is to illustrate the fact that $\kappa_{B}$ might be not equal to $\kappa$, it is necessary here to use the bound \eqref{eq:CMG_kappa} instead of the bound \eqref{eq:CMG_r}. Let $\lambda=5=\kappa + 2$. By using the Plotkin bound for $k_{\opt}^{(2)}$, we have
\[
k \leq 5 + k_{\opt}^{(2)}(13 - 2\cdot 6 + \GG(1,3), 3) = 5 + k_{\opt}^{(2)}(4,3) = 6.
\]
Hence, this code reaches the bound \eqref{eq:CMG_kappa}.

Finally, the binary code in Example \ref{ex:103323} has parameters $(10,3,3,2,3)$ and $\kappa=1$. By using the Plotkin bound, we get $\kappa_{B}=\kappa=1$ and $G(1,\delta)=3$. Let $\lambda=2=2\kappa$. We compute the bound \eqref{eq:CMG_r} using the Plotkin bound for $k_{\opt}^{(2)}$. We have
\[
k \leq 2 + k_{\opt}^{(2)}(10 - 2\cdot 3, 3) = 2 + k_{\opt}^{(2)}(4,3) = 3.
\]
Hence, the code is an optimal LRC.

Interestingly enough, to prove the optimality of both codes from Examples \ref{ex:199353} and \ref{ex:103323}, we used a set of entropy $k-1$ in the bounds \eqref{eq:CMG_kappa} and \eqref{eq:CMG_r}. However, none of the codes reaches the Singleton-type bound \eqref{eq:singl_G} obtained via a set of the same entropy. This is because we have an extra dependency on the field size by the bound $k_{\opt}^{(q)}$. Indeed there is no binary code with parameters $[4,2,3]$, while there is already an MDS code satisfying these parameters over $\F_{3}$. 

Finding a good family of LRCs achieving the bounds \eqref{eq:CMG_kappa} or \eqref{eq:CMG_r} is left for future work.

\section{Conclusion}
In this paper, we highlighted the fact that the parameter $r$ of an LRC with locality $(r,\delta)$ is only a loose upper bound on the number of nodes required for a local repair when the repair sets do not yield MDS codes. To remedy this, we proposed a new definition of locality called dimension-locality that restricts the dimension of the local codes instead of the size of the repair sets. We derived a new alphabet-dependent bound for LRCs of the same form as the bound in \cite{cadambe15}, which intrinsically uses the Griesmer bound. We showed that the bound is at least as good as the extension of the bound \cite{cadambe15} for LRCs with locality $(r,\delta)$ and the Singleton-type bound. We derived an asymptotic upper bound on the rate--distance tradeoff of LRCs and proved that the bound is the tightest known upper bound for large relative minimum distances. Finally, we derived the locality parameters of the Simplex codes for every local dimension to show that our bound is tight for every field size.

\section{Acknowledgment}

The authors are grateful to the anonymous reviewers and the associate editor for a number of very helpful comments that improved the presentation and quality of this paper. The first author would like to thank the Department of Electrical and Computer Engineering at the Technical University of Munich for hosting him while part of this work was carried out.

\bibliographystyle{IEEEtran}
\bibliography{IEEEabrv,referencesV2}


%

\appendix

%

\begin{proof}[Proof of Lemma \ref{lemma:correction}]
Remember that $H(R) = H(\cl(R))$ and $d_{C|\cl(R)} \geq \delta$. Now there exists $S$ in the \textit{res}-chain of $\cl(R)$ such that $H(S) - H(S \cap F) = \alpha$. Indeed if $S_{i}$ and $S_{i+1}$ are two consecutive sets in the \textit{res}-chain of $\cl(R)$, \ie,  $S_{i+1} \subset S_{i}$ and $H(S_{i+1})=H(S_{i})-1$, then we have
\begin{align*}
H(S_{1}) - H(F \cap S_{1}) & = H(S_{2}) + 1 - H(F \cap S_{1}) \\
& \leq H(S_{2}) + 1 - H(F \cap S_{2}).
\end{align*}

Therefore, the entropy difference between $S_{1}$ and its intersection with $F$ increases by at most $1$ compared to the entropy difference between $S_{2}$ and $F$. Since the last set in the \textit{res}-chain of $\cl(R)$ has entropy equal to $0$, we reach all possible integers in between $0$ and $H(\cl(R))-H(F \cap \cl(R))$ by evaluate every set in the \textit{res}-chain of $\cl(R)$.

By construction of the \textit{res}-chain of $\cl(R)$, we have $d_{S} \geq \left\lceil \dfrac{d_{\cl(R)}}{q^{H(R)-H(S)}} \right\rceil \geq \left\lceil \dfrac{\delta}{q^{H(R)-H(S)}}\right\rceil$. Thus, we can estimate the size of the added part via the Griesmer bound to have
\begin{align*}
|S|-|F \cap S| & \geq \GG(H(S)-H(F \cap S), d_{S}) \\
& \geq \GG \left( H(S)-H(F \cap S), \left\lceil \frac{\delta}{q^{H(R)-H(S)}}\right\rceil \right).
\end{align*}

Now since $H(S) \geq \alpha$, we have $H(R)-H(S) \leq \kappa - \alpha$. Hence we have

\[
\GG \left( H(S)-H(F \cap S), \left\lceil \frac{\delta}{q^{H(R)-H(S)}}\right\rceil \right) \geq \GG \left( \alpha, \left\lceil \frac{\delta}{q^{\kappa - \alpha}} \right\rceil  \right).
\]

Thus, by defining $F'=\cl(F \cup S)$, we have indeed
\begin{itemize}
\item $H(F') \leq H(F) + H(S) - H(F \cap S) = H(F) + \alpha$,
\item $|F'|\geq |F| + |S|- |F \cap S| \geq |F|  + \GG \left( \alpha, \left\lceil \dfrac{\delta}{q^{\kappa - \alpha}} \right\rceil \right) $. 
\end{itemize}
\end{proof}

\begin{proof}[Proof of Lemma \ref{lemma:recursive_correction}]
The idea of the proof is the following. We first add repair sets to $F$ as long as they do not increase the entropy too much. When there is one repair set that, when added, exceeds the bound on the entropy, we use Lemma \ref{lemma:correction} to add a smaller part and complete the set. We then prove that the corrected set has the desired size. 

We first define recursively the set that contains $F$ and some repair sets conditioned on the fact that there is no repair sets satisfying the hypotheses of Lemma \ref{lemma:correction}. The construction is the following.

\begin{enumerate}
\item Define $F_{0}=F$ and $\gamma_{0}=0$.
\item If for all repair sets $R$ we have $H(R)-H(F_{i-1}\cap R) < \kappa - \gamma_{i-1}$ then choose $R_{i}$ a repair set such that $\cl(R_{i}) \nsubseteq F_{i-1}$ and define
\begin{itemize}
\item $F_{i}=\cl(F_{i-1} \cup R_{i})$,
\item $\gamma_{i}=\gamma_{i-1} + H(R_{i}) - H(F_{i-1} \cap R_{i})$.
\end{itemize}
\end{enumerate}

The construction ends when there is at least one repair set with $H(R) - H(F_{i-1} \cap R) \geq \kappa - \gamma_{i-1}$. To see that the procedure always stops, assume for a contradiction that there is $j \in \Z_{+}$ such that $\{ R : R \text{ repair set } \} \subseteq F_{j}$. Since every coordinate $e \in [n]$ is contained in a repair set, we have $F_{j}=[n]$. But by construction, we have $H(R_{j}) - H(F_{j-1} \cap R_{j}) < \kappa - \gamma_{j-1}$. This implies that
\begin{align*}
k & = H(F_{j}) \leq H(F_{j-1}) + H(R_{j}) - H(F_{j-1} \cap R_{j}) \\
& < H(F_{j-1}) + \kappa - \gamma_{j-1} \\
& \leq H(F) + \gamma_{j-1} + \kappa - \gamma_{j-1} \\
& = H(F) + \kappa,
\end{align*}
which contradict the assumption that $H(F) + \kappa \leq k$. 

Let $F_{j}$ be the set obtain by the previous algorithm and let $R'$ the repair set such that $\cl(R') \nsubseteq F_{j}$ and $H(R')- H(R' \cap F_{j}) \geq \kappa - \gamma_{j}$. We first give an upper bound on the entropy and a lower bound on the size of $F_{j}$. By construction, we have $H(F_{j}) \leq H(F) + \gamma_{j}$. For the size of $F_{j}$, we will use the following approximation of $\GG(a,\delta) + \GG(b,\delta)$ for $a,b \in \Z_{+}$. We have
\[
\GG(a,\delta) + \GG(b,\delta) \geq \GG(a,\delta) + \GG \left( b,\left\lceil \frac{\delta}{q^{a}} \right\rceil \right) = \GG(a+b, \delta).
\] 
We can now give an estimation of the size of $F_{j}$.
\begin{align*}
|F_{j}| & \geq |F_{j-1}| + |R_{j}|-|F_{j-1} \cap R_{j}| \\
& \geq |F_{j-1}| + \GG(\gamma_{j}-\gamma_{j-1}, \delta) \\
& \geq |F| + \sum\limits_{i=1}^{j}\GG(\gamma_{i}-\gamma_{i-1}, \delta) \\
& \geq |F| + \GG(\sum\limits_{i=1}^{j} (\gamma_{i} - \gamma_{i-1}), \delta) \\
& \geq |F| + \GG(\gamma_{j}, \delta).
\end{align*}

Now we can apply Lemma \ref{lemma:correction} to $F_{j}$ by letting $\alpha = \kappa - \gamma_{j}$. Let $F_{c}$ be the set obtained from Lemma \ref{lemma:correction}. We check that $F_{c}$ has indeed the desired entropy and size. For the entropy, we have
\[
H(F_{c}) \leq H(F_{j}) + \kappa - \gamma_{j} \leq H(F) + \kappa.
\]
For the size of $F_{c}$, we have
\begin{align*}
|F_{c}| & \geq |F_{j}| + \GG \left( \kappa - \gamma_{j}, \left\lceil \dfrac{\delta}{q^{\gamma_{j}}} \right\rceil \right) \\
 & \geq |F| + \GG(\gamma_{j}, \delta) + \GG \left( \kappa-\gamma_{j}, \left\lceil \dfrac{\delta}{q^{\gamma_{j}}} \right\rceil \right) \\
  & = |F| + \GG(\kappa, \delta).
\end{align*}

Hence $F_{c}$ has both the required entropy and size and thus this concludes the proof. 
\end{proof}

\begin{proof}[Proof of Theorem \ref{thm:CMG_bound}]
Let $\lambda,a,b \in \Z_{+}$ such that $0 \leq \lambda \leq k$ and $\lambda = a\kappa + b$ with $0 \leq b < \kappa$. The proof is divided into three parts. First we prove that there exists a starting set $F_{s}$ with
\begin{itemize}
\item $H(F_{s}) \leq b$,
\item $|F_{s}| \geq \GG \left( b,\left\lceil \frac{\delta}{q^{\kappa-b}} \right\rceil \right)$.
\end{itemize}

To do this, we distinguish two cases. If there exists a repair set $R$ such that $H(R) \geq b$, then we can use Lemma \ref{lemma:correction} with $F=\cl(\emptyset)$ and $\alpha=b$ since $H(F)=0$. Define $F_{s}$ to be the set obtained from Lemma \ref{lemma:correction}. $F_{s}$ has indeed the required entropy and size. 

If there is no such repair set, it means that $b$ is a tighter upper bound on the entropy than $\kappa$. Therefore, $\CC$ is also an $(n,k,d)(b,\delta)$-LRC. By Lemma \ref{lemma:recursive_correction}, if $F = \cl(\emptyset)$, then there exists $F_{s}$ with $H(F_{s}) \leq b$ and $|F_{s}| \geq |\cl(\emptyset)| + \GG(b,\delta) \geq \GG \left( b, \left\lceil \frac{\delta}{q^{\kappa-b}} \right\rceil \right)$. 

Secondly, we iterate Lemma \ref{lemma:recursive_correction} $a$ times. Let $I$ be the final set. We have 
\begin{itemize}
\item $H(I) \leq H(F_{s}) + a\kappa \leq a\kappa + b = \lambda$,
\item $|I| \geq a\GG(\kappa, \delta) + |F_{s}| \geq (a+1)\GG(\kappa, \delta) - \GG(\kappa-b, \delta)$. 
\end{itemize}

Finally, now that we have the existence of this set, by the same argument as in \cite{cadambe15} with Lemma \ref{lemma:contract}, we can approximate the entropy of the shortened part by $k_{opt}$ with $n$ minus the bound on the size and the minimum distance $d$. Then, the best upper bound on $k$ is the minimum bound obtained over all $\lambda$. This concludes the proof. 
\end{proof}

\begin{proof}[Proof of Proposition \ref{prop:asympt_threshold}]

We use $\kappa_{a}$ as the upper bound on the dimension in \eqref{eq:CMG_Pltk_asympt} and let $\delta_{n}$ be such that $\delta_{n}\geq \delta_{t}$. We have
\begin{align*}
\delta_{n} &\geq \frac{1}{1+\frac{1}{q-1} \left( \dfrac{1}{1-\frac{\GG(\kappa_{a},\delta)}{r+\delta-1}} \right)} & \iff \\
\delta_{n} &\geq \frac{1}{1+\frac{r+\delta-1}{(q-1)(r+\delta -1 - \GG(\kappa_{a},\delta))}} & \iff \\
\delta_{n} &\geq \frac{1}{\frac{\frac{q}{q-1}(r+\delta-1) - \GG(\kappa_{a},\delta)}{r+\delta -1 - \GG(\kappa_{a},\delta)}} & \iff \\
\delta_{n} \kappa_{a} &\geq \kappa_{a} \frac{r+\delta -1 - \GG(\kappa_{a},\delta)}{\frac{q}{q-1}(r+\delta-1) - \GG(\kappa_{a},\delta)}
\end{align*}
Then, we have
\begin{flalign*}
&\delta_{n} \kappa_{a} \left(\GG(\kappa_{a},\delta) - \frac{q}{q-1}(r+\delta-1)\right) &\\
&\hspace{50pt}\leq \GG(\kappa_{a},\delta) \kappa_{a} - (r+\delta-1)\kappa_{a} &\iff \\ 
&\delta_{n}\kappa_{a} \frac{\GG(\kappa_{a},\delta) - \frac{q}{q-1}(r+\delta-1)}{\GG(\kappa_{a},\delta)(r+\delta-1)} &\\
&\hspace{50pt}\leq \frac{\GG(\kappa_{a},\delta) \kappa_{a} - (r+\delta-1)\kappa_{a}}{\GG(\kappa_{a},\delta)(r+\delta-1)} &\iff \\
&\delta_{n} \left( \frac{\kappa_{a}}{r+\delta-1} - \frac{\kappa_{a}}{\GG(\kappa_{a},\delta)}\frac{q}{q-1} \right) &\\
&\hspace{50pt}\leq \frac{\kappa_{a}}{r+\delta-1} - \frac{\kappa_{a}}{\GG(\kappa_{a},\delta)} &\iff \\ 
&\frac{\kappa_{a}}{\GG(\kappa_{a},\delta)} \left(1-\frac{\delta_{n}}{1-1/q}\right) 
 \leq \frac{\kappa_{a}}{r+\delta-1}(1-\delta_{n})
 &\\
\end{flalign*}
\end{proof}

\ifCLASSOPTIONcaptionsoff
  \newpage
\fi

\end{document}